\documentclass{article}
\usepackage[utf8]{inputenc}

\title{On Coding for an Abstracted Nanopore Channel for DNA Storage}
\author{Reyna Hulett\thanks{
Computer Science Department, Stanford University.  \texttt{rmhulett@stanford.edu}.
RH's research supported in part by a NSF Graduate Research Fellowship under grant DGE-1656518. },
\ Shubham Chandak\thanks{
Electrical Engineering Department, Stanford University.  \texttt{schandak@stanford.edu}.},
\ and Mary Wootters\thanks{
Computer Science and Electrical Engineering Departments, Stanford University.  \texttt{marykw@stanford.edu}.
RH and MW's research supported in part by NSF grants CCF-1844628 and SemiSynBio-1807371.}
}
\date{January 2021}

\usepackage[backend=bibtex]{biblatex}
\usepackage[margin=1in]{geometry}
\usepackage{blkarray}
\usepackage{amsmath}
\usepackage{amsfonts}
\usepackage{amsthm}
\usepackage{amssymb}
\usepackage{bm}
\usepackage{multirow}
\usepackage{thm-restate}
\usepackage{xcolor}
\usepackage{tikz}
\usepackage{comment}
\usepackage[ruled,vlined]{algorithm2e}
\usepackage{pgfplots}
\pgfplotsset{width=10cm,compat=1.9}
\usepgfplotslibrary{fillbetween}

\newtheorem{theorem}{Theorem}
\newtheorem{lemma}[theorem]{Lemma}

\theoremstyle{definition}
\newtheorem{definition}{Definition}

\newtheorem{remark}[theorem]{Remark}

\newcommand{\mkw}[1]{\textcolor{blue}{\textbf{[#1 --mary]}}}
\newcommand{\rh}[1]{\textcolor{green}{\textbf{[#1 --reyna]}}}
\newcommand{\TODO}[1]{\textcolor{red}{\textbf{[TODO: #1]}}}
\renewcommand{\mkw}[1]{} 
\renewcommand{\rh}[1]{} 
\renewcommand{\TODO}[1]{}

\newcommand{\eps}{\epsilon}
\begin{document}

\maketitle

\begin{abstract}
In the emerging field of DNA storage, data is encoded as DNA sequences and stored.  The data is read out again by sequencing the stored DNA.  Nanopore sequencing is a new sequencing technology that has many advantages over other methods; in particular, it is cheap, portable, and can support longer reads.  
While several practical coding schemes have been developed for DNA storage with nanopore sequencing, the theory is not well understood.  Towards that end,
we study a highly abstracted (deterministic) version of the nanopore sequencer, which highlights key features that make its analysis difficult.  We develop methods and theory to understand the capacity of our abstracted model, and we propose efficient coding schemes and algorithms.
\end{abstract}

\section{Introduction}
In the emerging field of \em DNA storage, \em data is encoded as DNA sequences and stored; the data can be read back by sequencing the stored DNA. 
This technology promises high storage density and stability, as well as efficient duplication of data and random access using PCR-based technologies; we refer the reader to \cite{ceze2019} and the references therein for an excellent overview. 

Both the synthesis and sequencing processes are noisy, and as a result the data must be encoded before the synthesis stage to ensure accurate data recovery.
Prior work has studied methods for encoding data in order to protect it against (aspects of) the noise introduced by these processes, for example~\cite{goldman2013,blawat2016,erlich,YGM17,organick2018,lee2019,lopez2019,chandak2020}. 

In this work, we focus on one particular stage of this noisy process, the \em nanopore sequencer. \em  
Nanopore sequencing---and in particular the MinION sequencer developed by Oxford Nanopore Technologies~\cite{jain2016oxford}---is an emerging sequencing technology.  While initial works in DNA 
storage used Illumina sequencing, nanopore sequencing has been attracting interest due to its portability, low cost, and ability to support significantly longer reads than Illumina.

At a high level,
the nanopore sequencer works as follows.  A single strand of DNA is passed through a pore, leading to variations in a current readout.  
The pore can hold $k$ nucleotides (for our purposes, a nucleotide is just a value in $\{A,C,G,T\}$) at a time; in practice $k$ is about six.  
The value of the current readout depends on which nucleotides are in the pore.  For example, if the strand of DNA is AGCTAAT, and the pore sees the sub-strand AGCT, it will output one current reading.  As the strand is passed through the pore, the contents of the pore will shift, say from AGCT to GCTA, then to CTAA, then to TAAT, and so on.  This will result in a change in the current reading, according to some function $f$ that maps $k$-mers to current levels.
The process is depicted in Figure~\ref{fig:pore}.  Given the current readout, the goal is to recover the original DNA sequence.

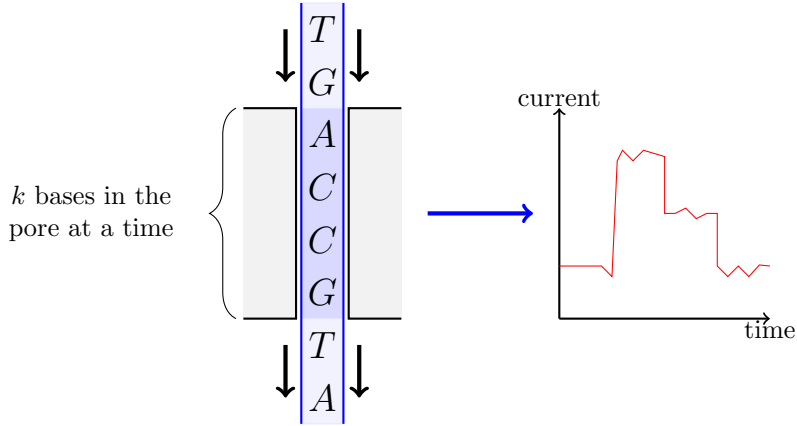
\begin{figure}
\begin{center}
\begin{tikzpicture}[scale=.7]

\draw[thick,fill=gray!10] (0,0) -- (1,0) -- (1,-4) -- (0,-4);
\draw[thick,fill=gray!10] (3,0) -- (2,0) -- (2,-4) -- (3,-4);
\draw [decorate,decoration={brace,amplitude=10pt},xshift=-4pt,yshift=0pt]
(0,-4) -- (0,0) node [black,midway,xshift=-0.3cm,anchor=east] {\begin{minipage}{3cm}\begin{center}$k$ bases in the \\pore at a time\end{center}\end{minipage}};
\draw[blue!10, fill=blue!5] (1.1, 2) rectangle (1.9, -6);
\draw[blue!10, fill=blue!15] (1.1, 0) rectangle (1.9, -4);
\draw[thick, blue] (1.1,2) to (1.1,-6);
\draw[thick,blue] (1.9, 2) to (1.9, -6);
\node at (1.5, 1.5) {\Large $T$};
\node at (1.5, .5) {\Large $G$};
\node at (1.5, -.5) {\Large $A$};
\node at (1.5, -1.5) {\Large $C$};
\node at (1.5, -2.5) {\Large $C$};
\node at (1.5, -3.5) {\Large $G$};
\node at (1.5, -4.5) {\Large $T$};
\node at (1.5, -5.5) {\Large $A$};
\draw[ultra thick,->] (.8, 1.5) to (.8, .5);
\draw[ultra thick,->] (2.2, 1.5) to (2.2, .5);
\draw[ultra thick,->] (.8, -4.5) to (.8, -5.5);
\draw[ultra thick,->] (2.2, -4.5) to (2.2, -5.5);

\draw[ultra thick, blue,->] (3.5, -2) -- (5.5,-2);

\begin{scope}[scale=2, xshift=3cm,yshift=-2cm]

\draw[->, thick] (0,0) to (2,0);
\draw[->, thick] (0,0) to (0,2);
\draw[red] (0,.5) -- (0.4,.5) -- (0.5, .4) -- (0.55, 1.5) -- (0.6, 1.6) -- (0.7, 1.5) -- (0.8, 1.6) -- (1,1.54) --(1,1)-- (1.1, 1) -- (1.2, 1.05) -- (1.3, .95) -- (1.4, 1) -- (1.5, 1) -- (1.5, .5) -- (1.6, .4) -- (1.7, .5) -- (1.8, .4) -- (1.9, .51) -- (2,.5);
\node at (2, -.1) {time};
\node at (0,2.1) {current};
\end{scope}
\end{tikzpicture}
\end{center}
\caption{High-level view of the nanopore sequencer.}\label{fig:pore}
\end{figure}

This channel is difficult to analyze, for several reasons.  
First, the output at any given time depends on $k > 1$ bases, and so 
there is inter-symbol interference.
Second, there may be collisions in the output: two different pore contents may lead to similar current readouts.  Third, the current readout can be noisy.  Fourth, the amount of time that each $k$-mer spends in the pore can vary, and sometimes never occur at all, leading to synchronization errors in the output.

Due to this complexity, typical basecallers (that is, methods for recovering the original sequence from the current readouts) rely on machine learning techniques~\cite{wick2019,flappie,scrappie}.  This is effective in practice, but difficult to get a theoretical handle on. 
While there are several practical approaches to 
error correction for nanopore sequencing in DNA storage~\cite{YGM17,organick2018,lee2019,lopez2019,chandak2020},
the theoretical limits of this channel are not well understood.  The work~\cite{MDK18}, discussed more below, proposed a probabilistic model of the nanopore sequencer and developed bounds on the capacity of the resulting channel.  However, this problem is quite difficult; in particular, taking into account the possibility of synchronization errors places this problem as more difficult than determining the capacity of the deletion channel, a notorious open problem.

In this work, we take a different approach to developing a theoretical understanding of the nanopore sequencer.  We develop a highly abstracted, \em deterministic \em model, which is meant to highlight the first two
sources of noise mentioned above: the inter-symbol interference, and the prospect of collisions when multiple $k$-mers lead to similar current outputs.
We develop methods and theory to understand the capacity of our abstracted channel, and we propose efficient coding schemes and algorithms.

Our goal in this work is to open up a research direction towards a theoretical understanding of the nanopore sequencer. 
We fully acknowledge that our work is not yet practical; in particular our abstracted channel does not include noise in the current readings or synchronization errors that can arise from variable pore dwelling times.  It is our hope that a solid understanding of our abstracted channel can then be combined with (much more well-studied) theories of coding for substitution and synchronization errors, in order to make progress on a more realistic channel model.

\textbf{Contributions.} Our contributions are as follows.
\begin{enumerate}
\item We propose a novel abstraction of the nanopore sequencer, which highlights the inter-symbol interference and the prospect of collisions.
This abstraction is simple enough that it is tractable, yet complex enough that it (a) captures some fundamental properties of the nanopore sequencer, and (b) already gives rise to extremely interesting problems from a theoretical perspective.
We hope that our abstraction will lead to future work in this area.
\item In Section~\ref{sec:compute}, we develop an algorithm to determine the capacity of this channel.  The algorithm is inefficient as pore size $k$ grows, but we can use it to determine the capacity for small $k$.
\item  In Section~\ref{sec:bound}, we develop simple bounds on the ``best'' and ``worst'' capacity for an arbitrary pore size $k$ (where ``best'' and ``worst'' refers to the choice of the map from $k$-mers to current readouts).  We use our algorithms mentioned above to compare these bounds to the exact values for $k=2$. These preliminary results suggest there may be some intriguing ``bumpiness'' in the worst-case capacity as the number of distinct current levels increases, but that the average-case is closer to best- than worst-case. 
\item In Section~\ref{sec:code}, we develop efficient coding schemes for our abstracted channel.  In particular, we develop a scheme that achieves rate $C_f - \epsilon$, where $C_f$ is the capacity, that requires preprocessing time $O\left( \frac{1}{\eps} \cdot 4^{1/\eps} \right)$, and has encoding/decoding time $O(n)$,
where the $O(\cdot)$ notation hides a constant that depends on $f$.
We also present an efficient coding scheme that improves over the ``naive'' one with high probability, where the probability is over a random map from $k$-mers to current readouts.
\end{enumerate}

\subsection{Related Work}
DNA storage has been around as an idea since the 1960's, and there has been renewed interest in it in the past decade, starting with the works~\cite{church, goldman2013}.
We refer the reader to \cite{ceze2019} for an excellent survey on DNA storage.  
Most works on DNA storage have focused on other sequencing technologies like Illumina, but
starting with the work of \cite{YGM17}, there have been several works which developed practical DNA storage systems for nanopore sequencers, including \cite{YGM17, organick2018, lee2019, lopez2019, chandak2020}.  All of these works developed practical coding schemes for DNA storage with a nanopore sequencer, but did not explicitly model the sequencer or analyze the theoretical limits.  

Perhaps the work most related to ours is that of \cite{MDK18}, who also developed a model of the nanopore sequencer and studied the capacity of their model.
In particular, they give a multi-letter capacity formula for their channel, and derive computable bounds for the capacity, in terms of a Markov transition matrix $P$ that captures the probability of transitioning from one $k$-mer to the next in the pore.
Our work complements that work by focusing on different aspects of the problem.
In more detail, that 
work differs from ours in several ways.
First, the model in \cite{MDK18} is stochastic, while ours is deterministic.  As a result, they take an information-theoretic approach, while our approach is more combinatorial in nature.
Second, their model includes the possibility that $k$-mers might get dropped; in particular, the binary deletion channel appears as one part of their model.  Since understanding the capacity of the binary deletion channel is a difficult open problem, this makes their problem extremely difficult. In contrast, we ignore this aspect in order to more cleanly focus in on the effects of the inter-symbol interaction and the potential confusability between the $k$-mers.  
Third, that work focuses on the nanopore sequencer for general applications (not necessarily for DNA storage), and in particular does not consider efficient coding schemes for their model.  Finally, that work derives bounds on the channel capacity for a particular choice of (a stochastic analog of) the map $f$ from $k$-mers to current levels, derived from experimental data.  In contrast, we are interested in results for any $f$, and in particular for the best and worst such functions $f$.  While the former direction is obviously of immediate interest for existing technology, it is our hope that understanding how the capacity of the channel changes with $f$ could perhaps guide how nanopore technology is developed in the future.  We note that this is still an emerging area and the technology is evolving; see~\cite{goto2020} for an overview of recent advances.

\section{Abstract Model of Nanopore Sequencer}

In this section we formalize our model.  As mentioned above, our goal is to focus on the challenges of (1) inter-symbol interference, and (2) the possibility of different $k$-mers producing similar current readouts.  With that in mind, we propose a very simple model for the nanopore sequencer.  As input, we take an encoded string $s_0 s_1 \cdots s_{n-1} \in \{A,C,G,T\}^n$.  This is transformed into a sequence of $k$-mers according to a sliding window, to obtain $(s_0 \cdots s_{k-1}), (s_1 \cdots s_k), (s_2 \cdots s_{k+1}), \ldots, (s_{n-k} \cdots s_{n-1})$.  Finally, each of these $k$-mers is mapped to one of $b$ distinct current levels, according to a mapping $f: \{A,C,G,T\}^k \to \{0,1, \ldots, b-1\}$.  This mapping $f$ defines the channel.

\begin{definition}[Abstract Nanopore Channel]
Given a mapping $f: \{A,C,G,T\}^k \to \{0,1,\dots,b-1\}$ from $k$-mers to current levels, 
let $f^*: \{A,C,G,T\}^* \to \{0,1,\dots,b-1\}^*$ represent the mapping from DNA strands to their full current readout, i.e.,
\[f^*(s_0 s_1 \cdots s_{n-1}) = f(s_0 \cdots s_{k-1}) \circ f(s_1 \cdots s_k)\circ \cdots \circ f(s_{n-k} \cdots s_{n-1}).\]
We call $f^*$ the \em abstract nanopore channel \em given by $f$.
\end{definition}

Given a mapping $f$, we are interested in the capacity $C_f$ (in bits-per-base), which we define as follows.

\begin{definition} 
Let $f: \{A,C,G,T\}^k \to \{0,1,\ldots, b-1\}$.  The capacity $C_f$ of the 
abstract nanopore channel determined by $f$ is defined as
\begin{equation}\label{eq:rf}
C_f = \lim_{n \to \infty} \frac{\log{|\{ c \in \{0,1,\dots,b-1\}^{n-k+1} \mid \exists s \in \{A,C,G,T\}^n \text{ s.t. } f^*(s)=c \}|}}{n}.
\end{equation}
\end{definition}
Observe that if $\mathcal{S}$ is a collection of 
strings $s \in \{A,C,G,T\}^n$ so that  $f^*(s)$ are all distinct for $s \in \mathcal{S}$, then by assigning a different message to each $s \in \mathcal{S}$, we can communicate perfectly (if not necessarily efficiently) across the abstract nanopore channel.

We will focus on \em balanced \em mappings $f$.  These are mappings so that $|f^{-1}(i)|$ is the same for all $i \in \{0, 1, \ldots, b-1\}$.

\section{Computing the Capacity}\label{sec:compute}

Our first contribution is an algorithm (Algorithm~\ref{alg:Rf} below) that computes  $C_f$, given $f$.  The basic idea is to consider
a finite automaton on the alphabet $\{0,1,\dots,b-1\}$ that accepts exactly those current readouts that can be generated by some DNA strand; then we use the transfer matrix method \cite{aydin2015, flajolet2009} for counting accepting paths in that finite automaton.

Formally, an Nondeterministic Finite Automaton (NFA) is a tuple $(Q, \Sigma, \Delta, Q_0, F)$ where $Q$ is the set of states, $\Sigma$ is the alphabet, $\Delta: Q \times \Sigma \to 2^Q$ is the transition function, $Q_0 \subseteq Q$ is the set of initial states, and $F \subseteq Q$ is the set of accepting states. Likewise, a Deterministic Finite Automaton (DFA) is a tuple $(Q, \Sigma, \delta, q_0, F)$, defined analogously except that $\delta: Q \times \Sigma \to Q$ is the transition function and there is only a single initial state $q_0$.

We consider the NFA $M$ and the DFA $M'$ described in Algorithm~\ref{alg:Rf}.  The NFA $M$ has states indexed by strings in $\{A,C,G,T\}^{k-1}$ and alphabet $\Sigma = \{0,1,\ldots,b-1\}$, the current levels.  Given a state $(s_0\cdots s_{k-2})$ and an input current level $i \in \Sigma$, the NFA $M$ can transition to any other state of the form $(s_1 \cdots s_{k-1})$ so that $f(s_0 s_1\cdots s_{k-1}) = i$.  All states are accepting states.  By construction, a sequence of current readings $c \in \Sigma^{n-k+1}$ can be an output $f^*(s_0s_1\cdots s_{n-1})$ of the nanopore sequencer if and only if $c$ is accepted by $M$.  The DFA $M'$ accepts exactly the same strings as $M$, and is obtained using the classic subset construction \cite{rabin1959}, such that each state of the DFA corresponds to a subset of the states of the NFA.

The \em transfer matrix method \em is a method for obtaining a \em generating function \em $g_f(z)$ 
for the number of states accepted by a given DFA.  In more detail, given the \em transfer matrix \em $T$ for the automaton (so that $T_{i,j}$ is the number of transitions from state $i$ to state $j$), the transfer matrix method gives an expression for a function $g_f(z)$ (in terms of the matrix $T$), so that 
\[ g_f(z) = \sum_{m=0}^\infty N_m z^m, \]
so that $N_m$ is the number of strings of length $m$ accepted by the finite automaton.  
We will use the notation $[z^m]g_f(z)$ to denote the coefficient $N_m$ on $z^m$.

\begin{lemma}[Transfer Matrix Method \cite{aydin2015}]\label{lem:tmm}
Given a DFA $D = (Q,\Sigma,\delta,q_0,F)$, let $D' = (Q \cup q_F, \Sigma \cup \lambda, \delta', q_0, \{q_F\})$ be obtained from the DFA $D$ by adding a new state $q_F$ and new symbol $\lambda$, along with $\lambda$-transitions from each of the accepting states of $D$ to $q_F$. Let $T$ be the transfer matrix of $D'$, where $T_{ij}$ is the number of transitions from state $i$ to state $j$, with $q_0$ being state $0$ and $q_F$ being state $|Q|$. Then the generating function for $D'$ is
\[g_f(z) = (-1)^{|Q|} \times \frac{\det(I-zT : |Q|, 0)}{z\det(I-zT)}\]
where $(I-zT : |Q|, 0)$ is the minor of index $|Q|,0$, i.e., the matrix $I-zT$ with the $|Q|^{th}$ row and $0^{th}$ column deleted.
\end{lemma}
In our case, the number of strings of length $m$ accepted by our finite automaton will be the number of current readouts of length $m = n-k+1$ that can be generated by some DNA strand of length $n$. 

From $g_f(z)$, we can derive the asymptotic behavior of the number of possible current readouts required to determine $C_f$.  As shown in the proof below, it is related to the smallest positive singularity of $g_f(z)$.

\begin{algorithm}[H]\label{alg:Rf}
\SetAlgoLined
\DontPrintSemicolon
 $Q \leftarrow \{A,C,G,T\}^{k-1}$\;
 $\Sigma \leftarrow \{0,1,\dots,b-1\}$\;
 NFA $M \leftarrow (Q, \Sigma, \Delta, Q, Q)$ where $\Delta(s_0 \cdots s_{k-2}, i) = \{s_1 \cdots s_{k-1} \mid f(s_0 s_1 \cdots s_{k-1}) = i\}$\;
 DFA $M' \leftarrow (2^Q, \Sigma, \delta, q_0, F)$ obtained from the subset construction on $M$\;
 $M' \leftarrow (2^Q \cup \{q_F\}, \Sigma \cup \{\lambda\}, \delta', q_0, \{q_F\})$ where \[\delta'(q, \sigma) = \begin{cases} 
      \delta(q, \sigma) & q \in 2^Q \text{ and } \sigma \neq \lambda \\
      q_F & q \in F \text{ and } \sigma = \lambda \\
      \emptyset & \text{otherwise}
   \end{cases}\]\;\vspace{-\baselineskip}
 $T \leftarrow$ transfer matrix of $M'$, i.e., $T_{i,j} =$ number of transitions from state $i$ to state $j$\;
 $g_f(z) \leftarrow (-1)^{2^{|Q|}} \times \frac{\det(I-zT : 2^{|Q|}, 0)}{z\det(I-zT)}$ where $(I-zT : 2^{|Q|}, 0)$ is the minor of index $2^{|Q|},0$\;
 $r \leftarrow$ smallest positive real root of the denominator of $g_f(z)$ (simplified)\;
 return $\log \frac{1}{r}$\;
 \caption{Calculate $C_f$}
\end{algorithm}

\begin{theorem}
Given a mapping $f:\{A,C,G,T\}^k \to \{0,1,\dots,b-1\}$, Algorithm~\ref{alg:Rf} computes $C_f$.
\end{theorem}
\begin{proof}
First, observe that the NFA $M$ accepts exactly those current readouts that can be generated by some DNA strand under the mapping $f$. For any current readout accepted by $M$, consider an accepting path $P = s_0\cdots s_{k-2}, s_1\cdots s_{k-1}, \dots, s_{m}\cdots s_{m+k-2}$. Then by construction, the DNA strand $s_0 \cdots s_{m+k-2}$ generates that current readout. In the other direction, for any DNA strand $s_0 \cdots s_{m+k-2}$, the path $P$ is an accepting path for $f^*(s_0 \cdots s_{m+k-2})$.

The DFA obtained from the subset construction accepts the same current readouts as $M$ \cite{rabin1959}. Then we apply the transfer matrix method described in Lemma~\ref{lem:tmm} to obtain the generating function $g_f(z)$, which counts the number of current readouts accepted by $M$.

Finally, we need to extract the asymptotic behavior of $[z^m] g_f(z)$. Since $g_f(z)$ is a generating function with non-negative coefficients, the Exponential Growth Formula \cite{flajolet2009} tells us that $\limsup_{m \to \infty} \left([z^m] g_f(z)\right)^{1/m} = 1/r$ where $r$ is the smallest positive singularity of $g_f(z)$. Note that the number of possible current readouts is monotonically non-decreasing in $m$, so the lim sup is equal to the limit. Therefore, based on equation~(\ref{eq:rf}), and the fact that the length of the current readouts for DNA strands of length $n$ is $m = n-k+1$,
\begin{align*}
C_f &= \lim_{n \to \infty} \frac{\log \left([z^{n-k+1}] g_f(z)\right)}{n}\\
&= \lim_{n \to \infty} \frac{n-k+1}{n} \log \left([z^{n-k+1}] g_f(z)\right)^{1/(n-k+1)}  \\
&= \log \frac{1}{r}
\end{align*}
\end{proof}

\subsection{Complexity of Algorithm~\ref{alg:Rf}}

Unfortunately, the DFA obtained from the subset construction has $2^{4^{k-1}}$ states, so the runtime of Algorithm~\ref{alg:Rf} is exponential in the problem size (i.e., the description length of $f$).

It is possible that calculating $C_f$ may be hard, because computing such a statistic for NFAs \emph{in general} is PSPACE-complete \cite{rampersad2009}. Specifically, even determining whether $C_f = \log b$ for $b=2$ or 4 ($C_f = \log b$ is the highest possible capacity over all mapping functions $f$; see Lemma~\ref{lem:bounds}) is equivalent to determining whether the corresponding NFA is universal (i.e., accepts every string in the alphabet).  We make this precise in the following lemma.

\begin{lemma}\label{lem:univ}
For $b=2$ or $4$ and any $f$, $C_f$ is equal to $\log b$ if and only if every current readout in $\{0,1,\dots,b-1\}^*$ can be generated by some DNA strand.
\end{lemma}
\begin{proof}
First we verify the ``if'' direction.  Indeed,
if every current readout can be generated, then $C_f = \lim_{n \to \infty} \frac{\log b^{n-k+1}}{n} = \log b$.

To show the ``only if'' direction we show the contrapositive. Suppose there is some current readout $c$ of length $\ell$ which cannot be generated by any DNA strand. Observe that any current readout which contains $c$ as a contiguous substring can also not be generated by any DNA strand (otherwise we could generate $c$ by truncation). Thus, consider any current readout of length $n-k+1$ which \emph{can} be generated by some DNA strand. If we divide it up into sections of length $\ell$, obviously none of those sections can be equal to $c$. So we see that
\begin{align*}
C_f &\leq \lim_{n \to \infty} \frac{\log \left((b^\ell -1)^{\lfloor \frac{n-k+1}{\ell} \rfloor} \ b^{(n-k+1) \bmod \ell}\right)}{n}\\
&= \lim_{n \to \infty} \frac{\lfloor \frac{n-k+1}{\ell} \rfloor \log (b^\ell -1)}{n}\\
&= \frac{1}{\ell} \log (b^\ell -1) < \log b.
\end{align*}
\end{proof}

Since universality is PSPACE-complete for general NFAs, an efficient algorithm would have to in some way leverage the highly structured nature of the NFAs corresponding to some mapping $f$.

\begin{remark}
Although we currently don't know how to calculate (or approximate) $C_f$ efficiently, it is possible to approximate the capacity for strands of fixed length $\ell$, $C_f(\ell)$. Specifically, there are existing randomized approximation algorithms for counting the number of strings of a given length accepted by an NFA, in time polynomial in $\ell$ and the size of the NFA \cite{kannan1995,arenas2019}.
\end{remark}

\section{Bounding the Capacity}\label{sec:bound}
The above approach for computing $C_f$ exactly given a mapping $f$ is only practical for small window sizes $k$. However, we can derive some general bounds that apply to any mapping $f$.  In particular, we are interested in the \em worst-case \em mapping $f$ (e.g., the one that attains $\min_f C_f$), as well as the \em best-case \em mapping $f$ (e.g., the one that attains $\max_f C_f$).  We prove the following lemma.

\begin{lemma}\label{lem:bounds}
For a given window size $k$ and with $b$ distinct current levels, we have the following bounds on $C_f$:
\begin{enumerate}
\item $\max_f C_f = \min(\log(b), 2)$\label{item:ub}
\item $\min_f C_f \geq \frac{\log(b)}{k}$
\item $\min_f C_f \leq 1$ when $b \leq 2^k$
\end{enumerate}
\end{lemma}
\begin{proof} We prove each bound independently:
\begin{enumerate}
\item Observe that the the size of the set in the numerator of equation~(\ref{eq:rf}) is bounded above both by the number of possible current readouts ($b^{n-k+1}$) and the number of DNA strands ($4^n$). Therefore $\max_f C_f \leq \min(\log(b), 2)$.

Furthermore, this rate is achievable. For $b=2$, it is achieved by, for instance,
$$f(s_0 \cdots s_{k-1}) = \begin{cases}
0 &\text{if $s_0 \in \{A,C\}$}\\
1 &\text{otherwise}
\end{cases}$$
since then all possible current readouts are achievable: given $c \in \{0,1,\dots,b-1\}^{n-k+1}$, simply replace all the 0's with $A$'s and the 1's with $G$'s, and add any $k-1$ bases to the end.

For $b \geq 4$, the optimal rate is achieved by any mapping where $f(s_0 \cdots s_{k-1}) = f(t_0 \cdots t_{k-1})$ implies $s_0 = t_0$, since then every DNA strand ending with $k-1$ A's gives rise to a unique current readout.

\item Regardless of the mapping $f$, we can always generate at least $b^{\lfloor n/k \rfloor}$ distinct current readouts: any choice of desired $0^{th}$, $k^{th}$, $2k^{th}$, etc. current readings can be obtained because they correspond to non-overlapping length-$k$ windows. Thus $\min_f C_f \geq \lim_{n \to \infty} \frac{\log b^{\lfloor n/k \rfloor}}{n} = \frac{\log(b)}{k}$.

\item When $b \leq 2^k$, we can effectively make the bases $A$ and $C$ indistinguishable from each other, and likewise the bases $G$ and $T$ indistinguishable from each other. Consider any balanced mapping $f':\{A,G\}^k \to \{0,1,\dots,b-1\}$ --- in this case, balanced means that for each $i$,
\[\left\lvert\{s_0 \cdots s_{k-1} \in \{A,G\}^k \mid f'(s_0 \cdots s_{k-1}) = i\}\right\rvert = 2^k/b.\]
Then define $f(s) = f'(s')$ where $s'$ is obtain from $s$ by replacing all $C$'s with $A$'s and all $T$'s with $G$'s. Observe that $f$ is a balanced mapping, and also the size of the set in the numerator of equation~(\ref{eq:rf}) is now at most $2^n$. Thus $\min_f C_f \leq \lim_{n\to\infty} \frac{\log(2^n)}{n} = 1$.
\end{enumerate}
\end{proof}

Combining these bounds, we obtain the plot in Figure~\ref{fig:p1}.
One might wonder how tight the bounds shown in Figure~\ref{fig:p1} are, and also about where $C_f$ lies for a ``typical'' mapping function $f$.
Using Algorithm~\ref{alg:Rf}, we are able to exactly calculate the maximum, minimum, and average $C_f$ for $k=2$, restricted to balanced mappings. These empirical results are shown in Figure~\ref{fig:p2}.

\begin{figure}
\begin{center}
\begin{minipage}{0.45\textwidth}
\centering
\begin{tikzpicture}
\begin{semilogxaxis}[
    log ticks with fixed point,
    xlabel={b},
    ylabel={bits-per-base},
    xmin=1, xmax=64,
    ymin=0, ymax=2,
    xtick={1,2},
    ytick={0,1,2},
    legend pos=north west,
    extra x ticks={4,8,16,32,64},
    extra x tick labels={\dots,$2^k$,$2^{k+1}$,\dots,$4^k$},
    width=\textwidth,
]
    
\addplot[
	name path=U,
    color=red,
    ]
    coordinates {
    (1,0)(2,1)(4,1)(8,1)(16,2)(32,2)(64,2)
    };
    
\addplot[
	name path=L,
    color=red,
    ]
    coordinates {
    (1,0)(64,2)
    };
    
\addplot [red!30] fill between [
of=U and L,
];

\addplot[
    color=blue,
    thick,
    dashed,
    ]
    coordinates {
    (1,0)(2,1)(4,2)(8,2)(16,2)(32,2)(64,2)
    };
\end{semilogxaxis}
\end{tikzpicture}

 \caption{Bounds on the value of $C_f$.  The dashed blue line is equal to $\max_f C_f$.
The value of $\min_f C_f$ must lie somewhere in the red shaded area.} \label{fig:p1}
\end{minipage}\hfill
\begin{minipage}{0.45\textwidth}
\centering
\begin{tikzpicture}
\begin{semilogxaxis}[
    log ticks with fixed point,
    xlabel={b},
    ylabel={bits-per-base},
    xmin=1, xmax=16,
    ymin=0, ymax=2,
    xtick={1,2,4,8,16},
    ytick={0,1,2},
    legend pos=north west,
    width=\textwidth,
]
    
\addplot[
	name path=U,
    color=red,
    ]
    coordinates {
    (1,0)(2,1)(4,1)(8,2)(16,2)
    };
    
\addplot[
	name path=L,
    color=red,
    ]
    coordinates {
    (1,0)(16,2)
    };
    
\addplot [red!30] fill between [
of=U and L,
];

\addplot[
    color=red,
    thick,
    ]
    coordinates {
    (1,0)(2,0.961914)(4,1)(8,1.771553)(16,2)
    };

\addplot[
    color=blue,
    thick,
    dashed,
    ]
    coordinates {
    (1,0)(2,1)(4,2)(8,2)(16,2)
    };
    
\addplot[
    color=green,
    thick,
    dotted,
    ]
    coordinates {
    (1,0)(2,0.999929)(4,1.837082)(8,1.989641)(16,2)
    };
    
\end{semilogxaxis}
\end{tikzpicture}
\caption{For $k=2$, the true $\max_f C_f$ (dashed blue), $\min_f C_f$ (solid red), and $\mathbb{E}_f C_f$ (dotted green), superimposed over our bounds from Figure~\ref{fig:p1}.}\label{fig:p2}
\end{minipage}
\end{center}
\end{figure}

The lower bound is on $\min_f C_f$ curious. Our theoretical bounds on $\min_f C_f$ are not tight, but the $k=2$ results suggest there may be some ``bumpiness'' in the true bound.
Also based on the $k=2$ results, it appears that random (balanced) mappings are closer to the best-case scenario than the worst-case. Section~\ref{sec:greedy} illustrates one coding scheme that takes advantage of a property shared by most random mappings for $b=2$, and we may hope there exist more general coding schemes that perform well for random mappings.

\section{Coding Schemes}\label{sec:code}
In the proof of Lemma~\ref{lem:bounds}, we observed that given a mapping $f$ with window size $k$ and $b$ distinct current levels, we can achieve a rate of $\frac{\log(b)}{k}$ by coding only on the $0^{th}, k^{th}, 2k^{th},$ etc. current readings. Here, we propose two generalizations of this trivial scheme to improve the rate with additional preprocessing based on the mapping $f$.

\subsection{``Block'' Encoding}
Instead of coding on non-overlapping windows of length $k$, each of which maps to one of $b$ current readings, we may instead choose a block length $\ell \geq k$ and code on non-overlapping blocks of length $\ell$. This requires precomputing an alphabet $\Sigma(\ell) \subseteq \{A,C,G,T\}^\ell$ such that $f^*\restriction_{\Sigma(\ell)}$ is an injective function with the same range as $f^*\restriction_{\{A,C,G,T\}^\ell}$. That is, each DNA strand in $\Sigma(\ell)$ generates a unique current readout, and together they generate every possible current readout of length $\ell-k+1$ given the mapping $f$. Provided our desired strand length $n$ is divisible by $\ell$ or tends to infinity, this coding scheme obtains the rate
\[C_f(\ell) = \frac{\log{|\{ c \in \{0,1,\dots,b-1\}^{\ell-k+1} \mid \exists s \in \{A,C,G,T\}^\ell \text{ s.t. } f^*(s)=c \}|}}{\ell} = \frac{\log |\Sigma(\ell)| }{\ell}\]
Since $\lim_{\ell \to \infty} C_f(\ell) = C_f$, it is possible to get arbitrarily close to the optimal rate by picking a sufficiently large $\ell$.

\begin{theorem}
Given a mapping $f$, for all $\epsilon > 0$, there is a coding scheme achieving rate $C_f - \epsilon$ with linear time encoding and decoding that requires preprocessing time $O(\frac{1}{\epsilon} \cdot 4^{1/\epsilon})$, where the $O(\cdot)$ notation hides constants that may depend on the mapping $f$.
\end{theorem}
\begin{proof}
We will exhibit such a scheme by choosing an appropriate block length $\ell$. Let $|\Sigma(\ell)|$ be the number of distinct current readouts that can be generated from DNA strands of length $\ell$. Equivalently, $|\Sigma(\ell)| = [z^{\ell-k+1}] g_f(z)$ is the number of current readouts of length $\ell-k+1$ accepted by the NFA $M$ constructed by Algorithm~\ref{alg:Rf}. Because $g_f(z)$ is a counting function for a regular language, and because $[z^{\ell-k+1}] g_f(z)$ is non-decreasing in $\ell$, the asymptotic behavior of $[z^{\ell-k+1}] g_f(z)$ has a simple form (\cite{flajolet2009}, Theorem V.3):
\[|\Sigma(\ell)| = [z^{\ell-k+1}] g_f(z) = \Theta(\Pi(\ell-k+1) (2^{C_f})^{\ell-k+1})\]
where $\Pi(x)$ is a polynomial.

Thus, there must exist some $\ell_0$ and some constant $C$ depending only on the mapping $f$, such that for all $\ell \geq \ell_0$, $|\Sigma(\ell)| \geq C (2^{C_f})^{\ell -k +1}$. Therefore, if we choose $\ell \geq \max\left(\ell_0, \frac{C_f \cdot (k-1) - \log C}{\epsilon}\right)$, we see that
\begin{align*}
C_f(\ell) &= \frac{\log |\Sigma(\ell)|}{\ell} \\
&\geq \frac{\log \left( C (2^{C_f})^{\ell-k+1} \right)}{\ell} \\
&= \frac{\log C + C_f \cdot (\ell-k+1)}{\ell}\\
&= C_f - \frac{C_f \cdot (k-1) - \log C}{\ell}\\
&\geq C_f - \epsilon.
\end{align*}
Therefore, to achieve rate $C_f - \epsilon$, we should choose $\ell$ proportional to $1/\epsilon$, with the constants depending only on the mapping $f$.

Given the block length $\ell$, we now describe how to compute the alphabet $\Sigma(\ell)$. We will construct an array $E$ containing the alphabet $\Sigma(\ell)$ and a hash table $D$ mapping current readouts of length $\ell-k+1$ to the index in $E$ of the DNA strand that generates that readout.

For each DNA strand $s$ of length $\ell$, compute $f^*(s)$. If $f^*(s)$ has not yet been added to the hash table $D$, append $s$ to the end of array $E$ and map $f^*(s)$ to the appropriate index. This preprocessing takes $O(\ell)$ time for each of $4^\ell$ DNA strands, for a total of $O(\ell \cdot 4^\ell) = O(\frac{1}{\epsilon} \cdot 4^{1/\epsilon})$.

Encoding and decoding are then quite straightforward: Convert the message to base $|\Sigma(\ell)|$ and use lookup table $E$ to map each digit to a block of length $\ell$. Similarly, decode each block of $\ell-k+1$ current readings using the hash table $D$ (skipping the $k-1$ readings that straddle each pair of adjacent blocks).
\end{proof}

As an example, consider the case when $b = 2$ or $4$ and $f$ is any mapping with the highest-possible capacity $C_f = \log b$. Per Lemma~\ref{lem:univ}, this implies that every current readout can be generated by some DNA strand, so $|\Sigma(\ell)| = b^{\ell-k+1}$ and $C_f(\ell) = \frac{(\ell-k+1)\log(b)}{\ell}$. In this case, we can calculate the dependence of $\ell$ on $\epsilon$ exactly:
\begin{align*}
C_f - \eps &= \frac{(\ell-k+1)\log(b)}{\ell} \\
- \eps &= \frac{(-k+1)\log(b)}{\ell}\\
\ell &= \frac{(k-1)\log(b)}{\eps}.
\end{align*}
For instance, when $k=2, b=2,$ and $\eps=0.1$, we would require $\ell=10$.

\subsection{``Greedy'' Encoding}\label{sec:greedy}
Alternatively, instead of changing the lengths of the blocks used in the trivial scheme, we may relax the ``non-overlapping'' requirement. 

This may not always be possible, depending on the mapping $f$. Indeed, in the worst case, it is possible that once you have fixed the first $k$ bases, the next $k-1$ current readings may also be fixed---for instance, if the first $k$ bases are all $A$, and all windows starting with $A$ map to the same current level. However, this pessimal case shouldn't happen for most mappings.

Consider the case of $b=2$ current levels, and suppose that for some length $1 \leq \ell < k$, for every ``prefix'' $p \in \{A,C,G,T\}^\ell$, at least one window in $f^{-1}(0)$ \emph{and} at least one window in $f^{-1}(1)$ starts with that prefix. Then it is possible to have every $(k-\ell)^{th}$ current reading code for one binary symbol independently. This would give us a rate of $\frac{1}{k-\ell}$ rather than $\frac{1}{k}$. Such an event is not too unlikely with a random mapping $f$.

\begin{lemma}
Given a random mapping $f$ with $b=2$ distinct current levels and a length $1 \leq \ell < k$, $f$ admits the coding scheme described above with rate $\frac{1}{k-\ell}$ with probability at least
\[1 - 4^\ell \cdot 2 \cdot \left(\frac{1}{2}\right)^{4^{k-\ell}}.\]
Furthermore,
\begin{enumerate}
\item we can determine whether such a scheme exists for a given $f$ and $\ell$ in $O(4^k)$.
\item we can find the maximum $\ell$ for which such a scheme exists for a given $f$ in $O(k 4^k)$.
\item we can implement such a scheme with $O(k 4^k)$ preprocessing and linear encoding and decoding.
\end{enumerate}
\end{lemma}
\begin{proof}
Let $f$ be a uniformly random mapping from $\{A,C,G,T\}^k \to \{0,1\}$ (not necessarily balanced). For a given prefix $p$ of length $\ell$ and current level $\beta$, the probability that all $4^{k-\ell}$ windows beginning with that prefix belong to $f^{-1}(\beta)$ is $\left(\frac{1}{2}\right)^{4^{k-\ell}}$. Thus by the union bound, the probability of being able to execute this scheme for a given length $\ell$ is at least
\[1 - 4^\ell \cdot 2 \cdot \left(\frac{1}{2}\right)^{4^{k-\ell}}.\]
This probability would only be higher if we restrict $f$ to be a random balanced mapping.

Given a specific mapping $f$ and desired length $\ell$, we can trivially check whether this property is satisfied by checking, for each prefix of length $\ell$, the current level corresponding to each of the $4^{k-\ell}$ windows beginning with that prefix, in time $O(4^k)$.

Naturally, we could determine the maximum $\ell$ that works for a given mapping $f$ in $O(k 4^k)$ by repeating the above procedure for each $\ell \in \{1,2,\dots,k-1\}$.

Furthermore, such a scheme could be straightforwardly implemented for a given $f$. For instance, we could construct tries --- trie 0 corresponding to $f^{-1}(0)$ and trie 1 corresponding to $f^{-1}(1)$ --- in $O(k 4^k)$. For ease of encoding, we will also equip each leaf node $s_0 s_1 \cdots s_{k-1}$ of the tries with two pointers, pointing to the internal nodes $s_{k-\ell} s_{k-\ell+1} \cdots s_{k-1}$ in both trie 0 and trie 1. To encode from binary to a DNA strand, we then start with any length-$k$ window in the trie corresponding to the first bit. For each subsequent bit, follow the pointer from our current leaf to the internal node of the trie corresponding to that bit, then append an arbitrary path from that prefix to a leaf in the appropriate trie. Decoding is simply a matter of extracting every $(k-\ell)^{th}$ current reading. Thus both encoding and decoding are linear in the length of the DNA strand.
\end{proof}

For example, with $k=6, \ell=4$, we see that we can obtain a rate of $1/2$ (compared to the trivial $1/6$) with probability at least $1 - \frac{1}{128}$.

\section{Conclusion}
We have initiated a theoretical study of coding for a highly abstracted version of the nanopore sequencer for DNA storage.   We have provided algorithms and bounds for understanding the capacity, and we have given efficient coding schemes.  However, we view our work as the tip of an iceberg.  First, even for this abstracted model, much remains open.  Can one derive better bounds on the capacity, or compute it efficiently for, say, $k=6$?  Second, we hope that our insights will generalize to more practical models, including with substitution and synchronization errors. 


\printbibliography


\appendix

\end{document}